\documentclass[12pt, draftclsnofoot, onecolumn]{IEEEtran}
\ifCLASSINFOpdf
\else
\fi
\usepackage{amsmath,amssymb,amsthm}

\usepackage{authblk}
\newtheorem{theorem}{Theorem}

\usepackage{hhline}
\usepackage{graphicx}
\usepackage{epstopdf}
\usepackage{subcaption}
\usepackage{algorithm}
\usepackage{algorithmic}
\usepackage{amsmath}
\usepackage{cite}

\begin{document}
	\title{Maximizing Indoor Wireless Coverage Using UAVs Equipped with Directional Antennas}
	\author{Hazim Shakhatreh and Abdallah Khreishah
		
		\thanks{Hazim Shakhatreh and Abdallah Khreishah are with the Department of
			Electrical and Computer Engineering, New Jersey Institute of Technology
			(email: \{hms35,abdallah\}@njit.edu)}}
	\maketitle
		\begin{abstract}
Unmanned aerial vehicles (UAVs) can be used to provide wireless coverage during emergency cases where each UAV serves as an aerial wireless base station when the cellular network goes down. They can also be used to supplement the ground base station in order to provide better coverage and higher data rates for the users. In this paper, we aim to maximize the indoor wireless coverage using UAVs equipped with directional antennas. We study the case that the UAVs are using one channel, thus in order to maximize the total indoor wireless coverage, we avoid any overlapping in their coverage volumes. We present two methods to place the UAVs; providing wireless coverage from one building side and from two building sides. In the first method, we utilize circle packing theory to determine the 3-D locations of the UAVs in a way that the total coverage area is maximized. In the second method, we place the UAVs in front of two building sides and efficiently arrange the UAVs in alternating upside-down arrangements. We show that the upside-down arrangements problem can be transformed from 3D to 2D and based on that we present an efficient algorithm to solve the problem. Our results show that the upside-down arrangements of UAVs, can improve the maximum total coverage by 100\% compared to providing wireless coverage from one building side.
		\end{abstract}
		
	\section{Introduction}
	Cells on wheels (COW), are used to provide expanded wireless coverage for short-term demands, when cellular coverage is either minimal, never present or compromised by the disaster~\cite{wikipedia}. UAVs can also be used to provide wireless coverage during emergency cases and special events (such as concerts, indoor sporting events, etc.), when the cellular network service is not available or it is unable to serve users~\cite{bupe2015relief,bor2016efficient,haz2017efficient,hazim2017}. Compared to the COW, the advantage of using UAV-based aerial base stations is their ability to quickly and easily move~\cite{mozaffari2016mobile}. The main disadvantage of using UAVs as aerial base stations is their energy capacity, the UAVs need to return periodically to a charging station for recharging, due to their limited battery capacity. In~\cite{shakhatreh2016continuous}, the authors integrate the recharging requirements into the coverage problem and examine the minimum number of required UAVs for enabling continuous coverage under that setting.
	
	Directional antennas are used to improve the received signal at their associated users, and also reduce interference since other- aerial base stations are targeting/serving other users in other directions~\cite{georgiou2016simultaneous}. The authors in~\cite{mozaffari2016efficient} study the optimal deployment of UAVs equipped with directional antennas, using circle packing theory. The 3D locations of the UAVs are determined in a way that the total coverage area is maximized.
	In~\cite{azari2016optimal}, the authors investigate the problem by characterizing the coverage area for a target outage probability, they show that for the case of Rician fading there exists a unique optimum height that maximizes the coverage area. In~\cite{kalantari2016number}, the authors propose a heuristic algorithm to find the positions of aerial base stations in an area with different user densities, the goal is to find the minimum number of UAVs and their 3D placement so that all the users are served. However, it is assumed that all users are outdoor and the location of each user represented by an outdoor 2D point. In~\cite{orfanus2016self}, the authors use multiple UAVs to design efficient UAV relay networks to support military operations. They describe the tradeoff between connectivity among the UAVs and maximizing the covered area. However, they use the UAVs as wireless relays and do not take into account their mutual interference in downlink channels. In~\cite{kosmerl2014base}, the authors propose a computational method for positioning aerial base stations with the goal of minimizing their number, while fully providing the required bandwidth over the disaster area. It is assumed that overlapping aerial base stations coverage areas are allowed and they use the Inter-Cell Interference Coordination (ICIC) methods to schedule radio resources to avoid inter-cell interference. The authors in~\cite{haz2017efficient,hazim2017} use a single UAV equipped with omnidirectional antenna to provide wireless coverage for indoor users inside a high-rise building, where the objective is to find the 3D location of a UAV that minimizes the total transmit power required to cover the entire high-rise building. In~\cite{haz2017dir}, the authors
	use UAVs equipped with omnidirectional antennas to minimize the number of UAVs required to cover the indoor users.
	
	We summarize our main contributions as follows:
	\begin{itemize}
		\item In order to maximize the indoor wireless coverage, we present two methods to place the UAVs, providing wireless coverage from one building side and from two building sides. In this paper, we study the case that the UAVs are using one channel, thus we avoid any overlapping in their coverage volumes (to avoid interference). In the first method, we utilize circle packing theory to determine the 3-D locations of the UAVs in a way that the total coverage area is maximized. In the second method, we place the UAVs in front of two building sides and efficiently arrange the UAVs in alternating upside-down arrangements.
		\item We show that the upside-down arrangements problem can be transformed from 3D to 2D and based on that we present an efficient algorithm to solve the problem.
		\item We demonstrate through simulation results that the upside-down arrangements of UAVs, can improve the maximum total coverage by 100\% compared to providing wireless coverage from one building side.
	\end{itemize}	
	The rest of this paper is organized as follows. In Section II,
	we describe the system model. In Section III, we show the appropriate placement of UAVs that maximizes the total indoor wireless coverage. Finally, we present our numerical results in Section IV and make concluding remarks in Section V.
	\section{System Model}
	\label{sec:system_model}
	\subsection{System Settings}
	\label{subsec:system settings}
	Consider a $3D$ building, as shown in Figure~\ref{fig1}, where $N$ UAVs must be deployed to maximize wireless coverage to indoor users located within the building. The dimensions of the high-rise building, in the shape of a rectangular prism, be $[0,x_b]$ $\times$ $[0,y_b]$ $\times$ $[0,z_b]$. Let ($x_{k}$, $y_{k}$, $z_{k}$) denote the 3D location of UAV $k$$\in$ $N$, and let ($X_{i}$, $Y_{i}$, $Z_{i}$) denote the location of user $i$.  Also, let $d_{out,i}$ be the distance between the UAV and indoor user $i$, and let $d_{in,i}$ be the distance between the building wall and indoor user $i$. Each UAV uses a directional antenna to provide wireless coverage where the antenna half power beamwidth is $\theta_{B}$. The authors in~\cite{kaya2016wireless} use an outdoor directional antenna to provide wireless coverage for indoor users. They show that the highest RSRP (Reference Signal Received Power) and throughput values are measured along the main beam direction, thus the radiation pattern of a directional antenna is a cone and the indoor volume covered by a UAV is a truncated cone, as shown in Figure~\ref{fig2}. Here, $r_{i}$ is the radius of the circle that is located at $yz$-rectangular side ((0,0,0), (0,0,$z_b$) , (0,$y_b$,$z_b$), (0,$y_b$,0))), $r_{j}$ is the radius of the circle that is located at $yz$-rectangular side (($x_b$,0,0), ($x_b$,0,$z_b$) , ($x_b$,$y_b$,$z_b$), ($x_b$,$y_b$,0)) and $x_b$ is the horizontal width of the building. The volume of a truncated cone is given by:
	\begin{equation*}
	\begin{split}
	V=\frac{1}{3}\pi x_{b} (r^2_{i}+r^2_{j}+r_{i}r_{j})
	\end{split}
	\end{equation*}  
	\begin{figure}[t]
		\centering
		\includegraphics[scale = 0.65]{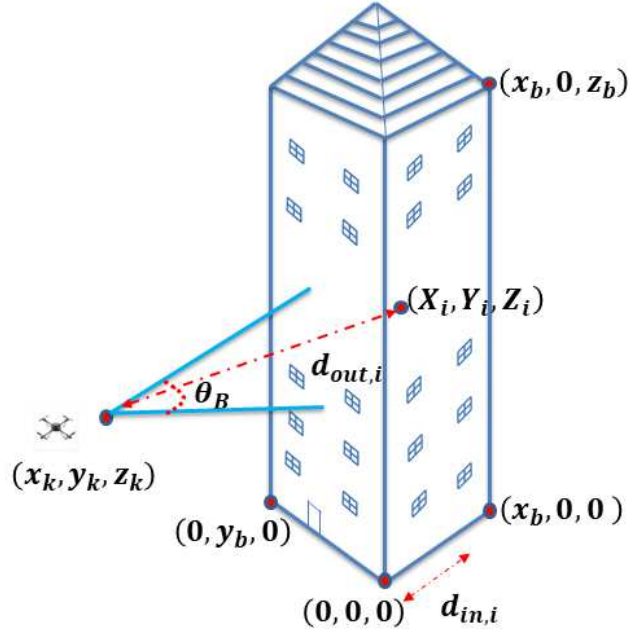}
		\caption{System model}
		\label{fig1}
	\end{figure}
		\subsection{UAV Power Consumption}
		\label{subsec:optimal}	
		In~\cite{feick2015achievable}, the authors show that significant power gains are attainable for indoor users even in rich indoor scattering conditions, if the indoor users use directional antennas. Now, consider a transmission between $k$-th UAV located at
		($x_{k}$, $y_{k}$, $z_{k}$) and $i$-th indoor user located at ($X_{i}$, $Y_{i}$, $Z_{i}$). The received signal power at $i$-th indoor user location can be given by:
		\begin{equation*}
		\begin{split}
		P_{r,ik}(dB)=P_t+G_t+G_r-L_i
		\end{split}
		\end{equation*}  
		where $P_{r,ik}$ is the received signal power, $P_t$ is the transmit power of UAV, $G_t$ is the antenna gain of the UAV. It can be approximated by $G_t\approx\frac{29000}{\theta_B^2}$ with $\theta_{B}$ in degrees~\cite{venugopal2016device,balanis2016antenna} and $G_r$ is the antenna gain of indoor user $i$, which is given by~\cite{feick2015achievable}:
		\begin{equation*}
		\begin{split}
	    G_r(dB)=G_{r,dir}+G_{r,omni}-GRF
		\end{split}
		\end{equation*}  
		where $G_{r,dir}$ and $G_{r,omni}$ are free-space antenna gains of a directive and an omnidirectional antenna respectively and $GRF$ is the decrease in gain advantage of a directive over an omnidirectional antenna, due to the presence of clutter.
		
		 Also, $L_i$ is the path loss which for the Outdoor-Indoor communication is:
		\begin{equation*}
		\begin{split}
		L_i=L_{F}+L_{B}+L_{I}= 
		(w\log_{10}d_{3D,i}+w\log_{10}f_{Ghz}+g_{1})\\+
		(g_{2}+g_{3}(1-\cos\theta_{i})^{2})+(g_{4}d_{2D,i})
		\end{split}
		\end{equation*}
		where $L_{F}$ is the free space path loss, $L_{B}$ is the building penetration loss, and $L_{I}$ is the indoor loss. In the path loss model, we also have
		$w$=20, $g_{1}$=32.4, $g_{2}$=14, $g_{3}$=15, $g_{4}$=0.5~\cite{series2009guidelines} and $f_{Ghz}$ is the carrier frequency.
			\begin{figure*}[t]
				\begin{minipage}[b]{0.5\linewidth}
					\centering
					\includegraphics[width=\textwidth]{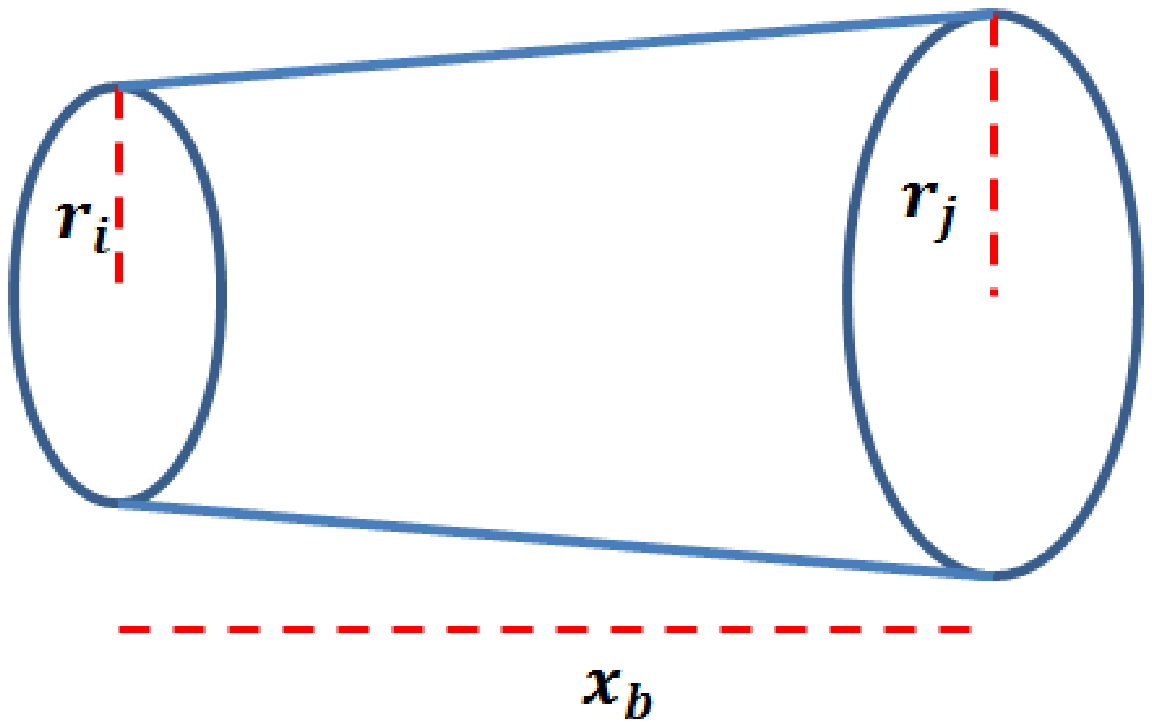}
					\caption{3D Dimensions of the truncated cone
					}
					\label{fig2}
				\end{minipage}
				\hspace{2cm}
				\begin{minipage}[b]{0.4\linewidth}
					\centering
					\includegraphics[width=\textwidth]{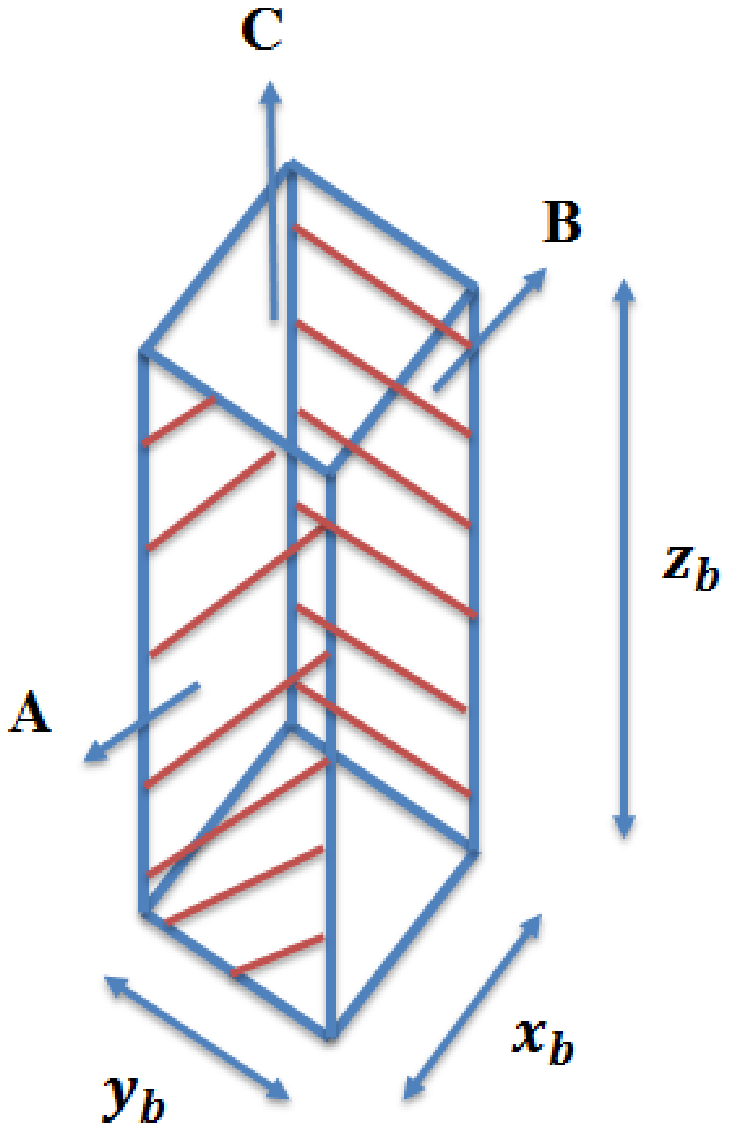}
					\caption{Building sides
					}
					\label{fig3}
				\end{minipage}
			\end{figure*}
		\subsection{Placement of UAVs } 
		Choosing the appropriate placement of UAVs will be a
		critical issue when we aim to maximize the indoor wireless coverage. In this paper, we assume that we can place the UAVs in front of building sides $A$, $B$ and above the building $C$ as shown Figure~\ref{fig3}. We also assume that the UAVs are using one channel. In this section, we demonstrate why avoiding the overlapping between UAV's coverage volumes will strengthen the total indoor wireless coverage. 
		\begin{figure*}[!t]
			\centering
			\includegraphics[scale=0.32]{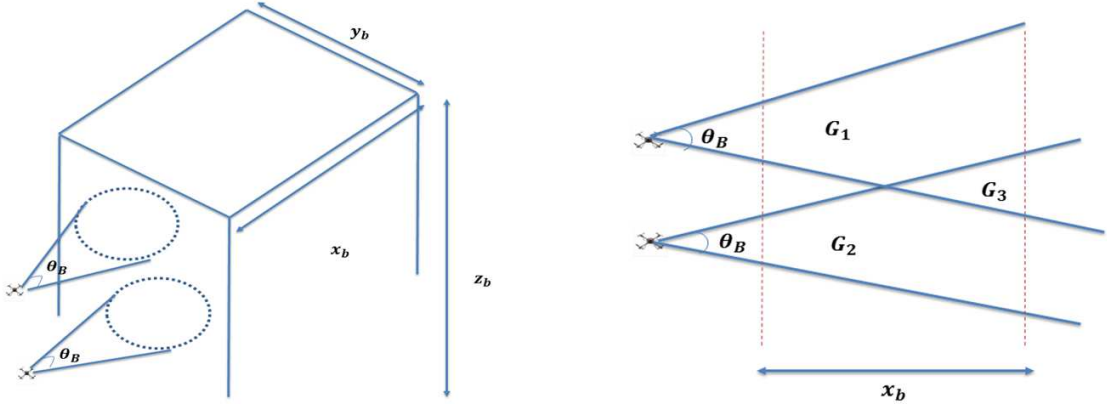}
			\caption{Placing two UAVs in front of building side A}
			\label{fig4}
		\end{figure*}
		\begin{figure*}[!t]
			\centering
			\includegraphics[scale=0.32]{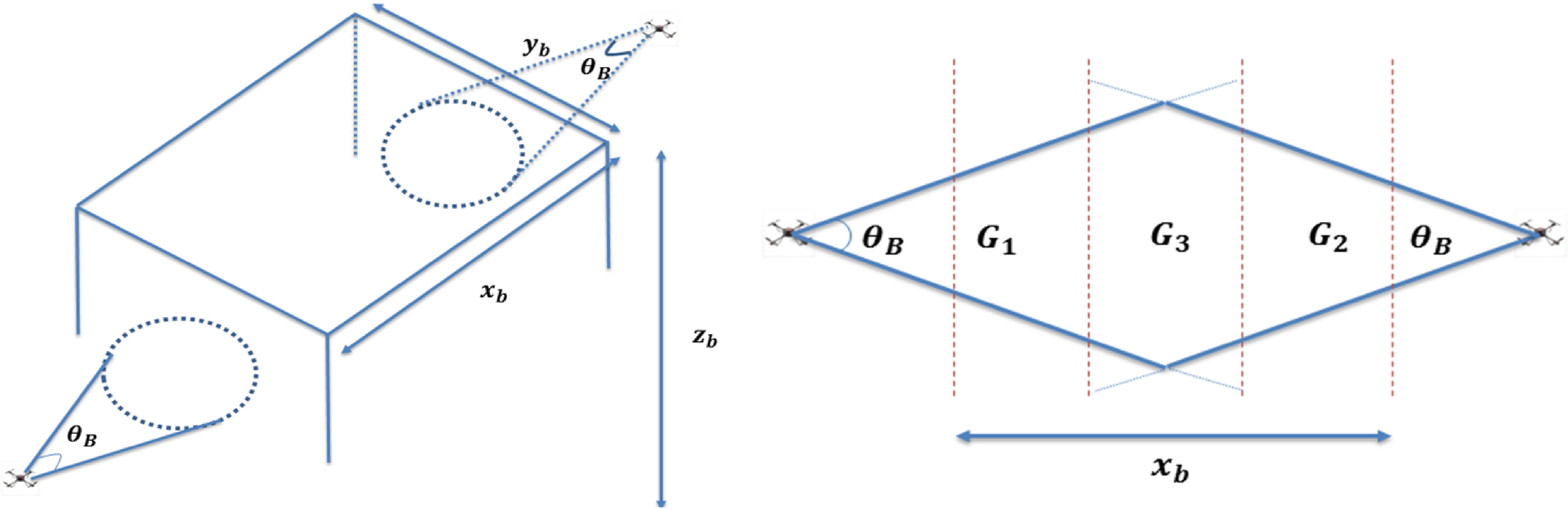}
			\caption{Placing two UAVs in front of two building sides A and B}
			\label{fig5}
		\end{figure*}
		\subsubsection{Overlapping between UAV's coverage volumes is allowed}
		 Now, when we place two UAVs in front of building sides $A$ as shown in Figure~\ref{fig4} (the UAVs have different z-coordinates and same x- and y- coordinates), the indoor users located in $G_1$'s and $G_2$'s locations will have high $SINR$. On the other hand, the indoor users located in $G_3$'s location will have low $SINR$. This is because the dependency of SINR on the location of indoor user. Similarly, when we place two UAVs in front of two building sides $A$ and $B$ as shown in Figure~\ref{fig5} (the UAVs have different x-coordinates and same y- and z- coordinates), the indoor users located in $G_1$'s and $G_2$'s locations will have high $SINR$. On the other hand, the indoor users located in $G_3$'s location will have low $SINR$. In Figure~\ref{fig6} (the UAVs have same y-coordinates and same x- and z- coordinates), when we place one UAV in front of building side A and
		 one UAV above the building C, the indoor users located in $G_1$'s and $G_2$'s locations will have high $SINR$. On the other hand, the indoor users located in $G_3$'s location will have low $SINR$. From the previous examples, we can conclude that allowing the UAVs coverage volumes to overlap will result in that some users are not satisfied. In the next section, we place the UAVs in a way that maximizes the total coverage, and avoids any overlapping in their coverage volumes.
		 \begin{figure*}[!t]
		 	\centering
		 	\includegraphics[scale=0.33]{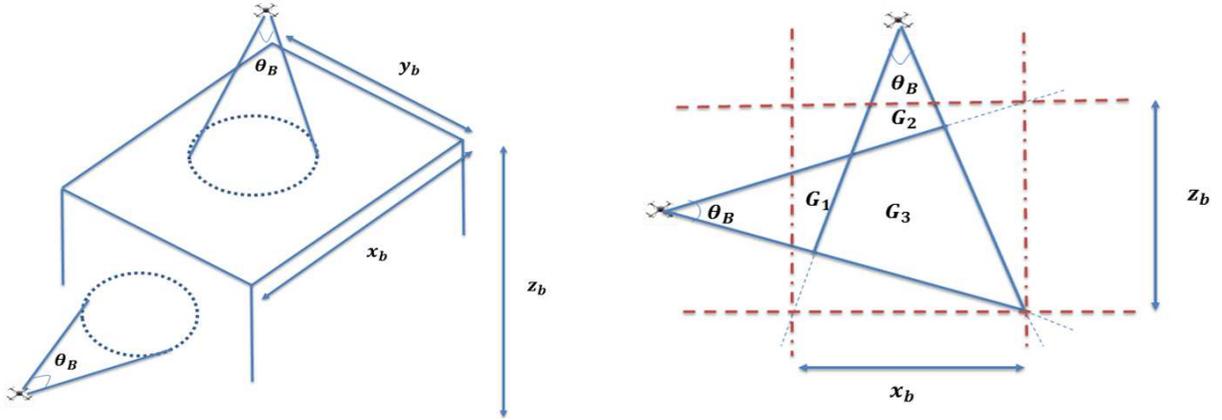}
		 	\caption{Placing one UAV in front of building side A and
		 		one UAV above the building C}
		 	\label{fig6}
		 \end{figure*}
		 \begin{figure*}[!t]
		 	\centering
		 	\includegraphics[scale=0.33]{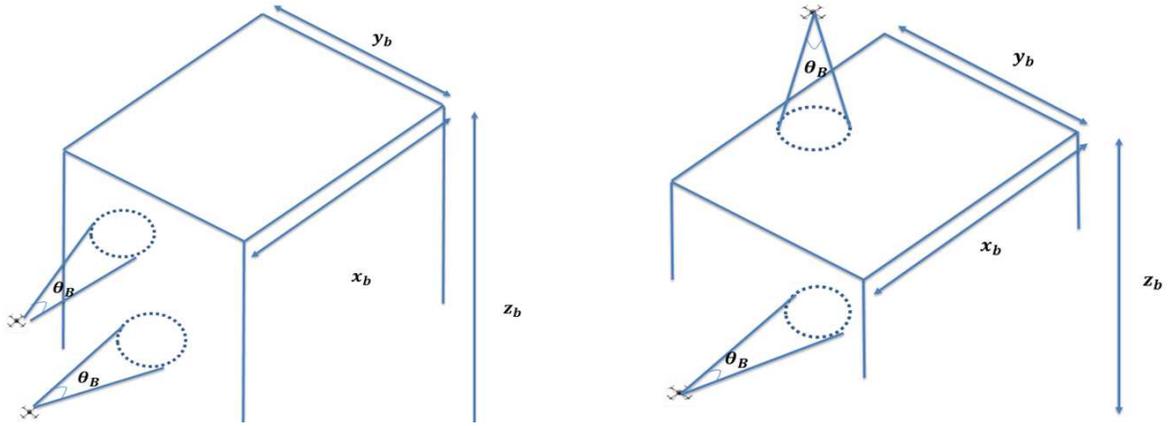}
		 	\caption{UAVs with small antenna half power beamwidth $\theta_{B}$}
		 	\label{fig7}
		 \end{figure*}
		 \subsubsection{Overlapping between UAV's coverage volumes is not allowed}
		 In Figure~\ref{fig7}, we avoid the overlapping between UAV's coverage volumes by using UAVs with small antenna half power beamwidths $\theta_{B}$. Actually, this is impractical way to cover the building, due to the high number of UAVs required to cover the building. In Figure~\ref{fig8}, we place the UAVs in front of two building sides and efficiently arrange the UAVs in alternating upside-down arrangements. We can notice that this method will maximize the indoor wireless coverage where the uncovered holes are minimized and the overlapping between UAV's coverage volumes is not allowed.   
			\section{Maximizing Indoor Wireless Coverage}
			In this section, the UAVs are assumed to be symmetric having the same transmit power, the same horizontal location $x_{k}$, the same channel and the same antenna half power beamwidth $\theta_{B}$. We show two methods to place the UAVs in a way that tries to maximize the total coverage, and avoids any overlapping in their coverage volumes.
	\subsection{Providing Wireless Coverage from one building side}
	
	 In this method, we place all UAVs in front of one building side (side $A$, side $B$ or side $C$). The objective is to determine the three-dimensional location of each UAV $k$$\in$ $N$ in a way that the total covered volume is maximized. Now, consider that we place the UAVs in front of building side $A$, then the projection of UAV's coverage on the building side $B$ is a circle as shown in Figure~\ref{fig9}. Our problem can be formulated as:
		\begin{equation*}
		\begin{split}
		\max~~ |N|\star\frac{1}{3}\star\pi\star x_{b}\star (r^2_{i}+r^2_{j}+r_{i}r_{j})~~~~~~~~~~~~~~~~~~~~~\\
		subject~to~~~~~~~~~~~~~~~~~~~~~~~~~~~~~~~~~~~~~~~~~~~~~~~~~~~~~~~~~~~~\\
		\sqrt{(y_k-y_q)^2+(z_k-z_q)^2}\geq 2r_{j},~~k\neq q\in N\\
		z_b-(z_{k}+r_{j})\geq 0,~~k\in N~~~~~~~~~~~~~~~~~~~~~\\
		(z_{k}-r_{j})\geq 0,~~k\in N~~~~~~~~~~~~~~~~~~~~~\\
		y_b-(y_{k}+r_{j})\geq 0,~~k\in N~~~~~~~~~~~~~~~~~~~~~\\
		(y_{k}-r_{j})\geq 0,~~k\in N~~~~~~~~~~~~~~~~~~~~~\\
		\end{split}
		\end{equation*}
	
	The objective is to maximize the indoor wireless coverage (covered volume). Constraint set (1) guarantees that truncated cones cannot overlap each other. Constraint sets (2-5) ensure that UAV $k$ should not cover outside the $3D$ building, see Figure~\ref{fig9}. We model this problem by utilizing the well-known problem, circle packing problem. In this problem, $N$ circles should be packed inside a given surface such that the packing density is maximized and no overlapping occurs~\cite{birgin2005optimizing}, note that the surface in our problem is a rectangle. 
	
    The authors of~\cite{birgin2005optimizing} tackle this problem by solving a number of decision problems. The decision problem is:\\
    \\
    \textit{Given $N$ circles of radius $r_j$ and a rectangle of dimension $d_1\times d_2$, whether is it possible to locate all the circles into the rectangle or not.}\\
    \begin{figure*}[!t]
    	\centering
    	\includegraphics[scale=0.32]{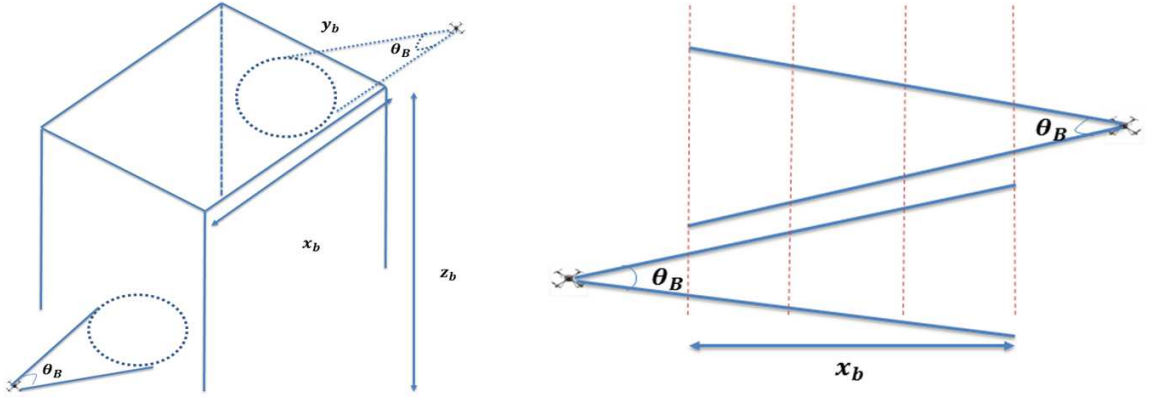}
    	\caption{UAVs in alternating upside-down arrangements}
    	\label{fig8}
    \end{figure*}
    \begin{figure}[!t]
    	\centering
    	\includegraphics[scale = 0.85]{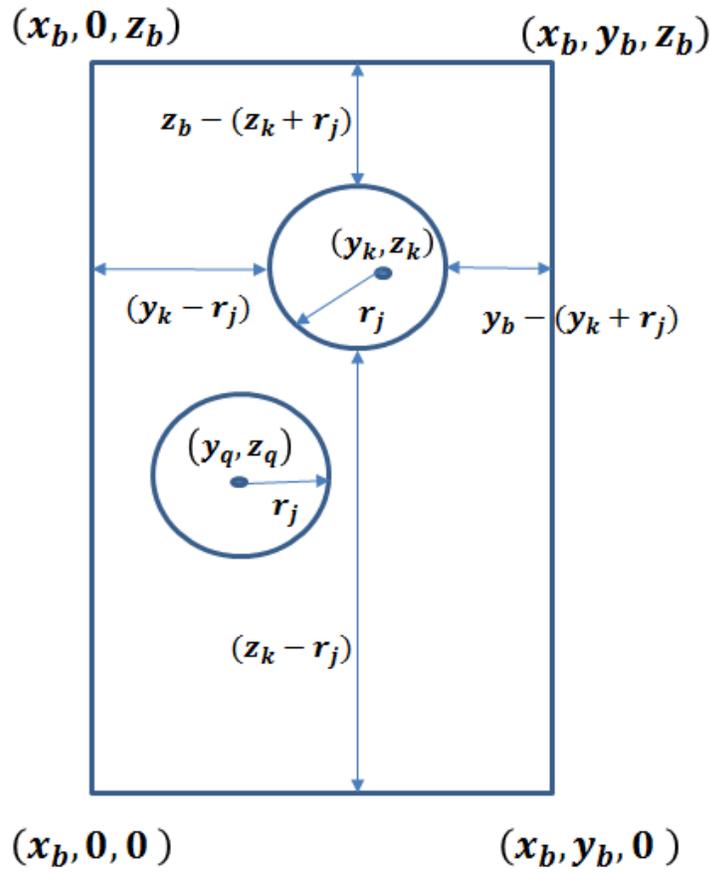}
    	\caption{Circle packing in a rectangle}
    	\label{fig9}
    \end{figure}
    
    In~\cite{birgin2005optimizing}, the authors introduce a nonlinear model for this problem. Finding the answer for the decision problem will depend on finding the global minimizer of a nonconvex and nonlinear optimization problem. In each decision problem, they investigate the feasibility of packing $N$ identical circles. If this is feasible, $N$ is incremented by one and the decision problem is solved again. The algorithm will stop when the decision problem yields an infeasible packing~\cite{hifi2009literature}. The pseudo code of the algorithm is shown in Algorithm 1. In the next section, we utilize the two building sides to maximize the indoor wireless coverage. This will allow us to extend the indoor wireless coverage compared with providing wireless coverage from one building side, because the holes induced by the cones of the UAVs in one side can be filled by the cones induced by the UAVs in the other side without causing overlap among the two sets of cones.
    \begin{algorithm}
    	\begin{algorithmic}
    		\STATE 1: $N$ $\longleftarrow$ 1 
    		\STATE 2: Solve the decision problem for $N$ circles
    		\STATE 3: If Answer $=$ YES
    		\STATE 4: \textbf{Then} $N\longleftarrow N+1$
    		\STATE 5: \textbf{Return} to step 2
    		\STATE 6: If Answer $=$ NO
    		\STATE 7: $n\longleftarrow N-1$
    		\STATE 8: \textbf{End}
    		\STATE 9: \textbf{Output} $n$
    	\end{algorithmic}
    	\caption{Circle packing in a rectangle}
    \end{algorithm}
		\subsection{Providing Wireless Coverage from two building sides}	
		In this method, we place the UAVs in front of two building sides (side $A$ and side $B$) and
		efficiently arrange the UAVs in alternating upside-down arrangements (see Figures~\ref{fig10} and ~\ref{fig11}). In Theorem 1, we find the horizontal location of the UAV $x_{UAV}$ that guarantees the upside-down arrangements of the truncated cones. In Theorem 2, we prove that if the truncated cones do not intersect in 3D, then the circles do not intersect in building sides (A and B), and vice versa. In Theorem 3, we prove that if we maximize the percentage of covered area of building sides (A and B), then we maximize the percentage of covered volume of building, and vice versa. These theorems help us to transform the geometric problem from 3D to 2D and present an efficient algorithm that maximizes the indoor wireless coverage.
		\begin{theorem}
			The horizontal location of the UAV $x_{UAV}$ that guarantees the upside-down arrangements of the truncated cones will be equal to $0.7071x_b$ regardless of the antenna half power beamwidth angle $\theta_B$.
		\end{theorem}
		\begin{proof}
			The radius of the smaller circular face $r_i$ is given by:
			
			\begin{equation}
			r_i=r_j \dfrac{x_{UAV}}{x_b+x_{UAV}}
			\end{equation}
			Now, we divide the building sides $A$ and $B$ to square cells (as shown in Figures~\ref{fig10} and~\ref{fig11}), the large circle in Figure~\ref{fig10} 
			and the small circle in Figure~\ref{fig11} will represent the projections of UAV's coverage on building sides $A$ and $B$ when the UAV is placed in front of building side $B$. Similarly, the four small circles quarters in Figure~\ref{fig10} and
			the four large circles quarters in Figure~\ref{fig11} will represent the projections of UAVs coverage on building sides $A$ and $B$ when the UAV is placed in front of building side $A$. From Figures~\ref{fig10} and ~\ref{fig11}, the diagonal of the square cell is:
			\begin{equation*}
			 D=2r_j+2r_i
			 \end{equation*}
			 where $r_j$ is the radius of the larger circular face and $r_i$ is the radius of the smaller circular face. By applying the pythagorean’s theorem, we get:
			
			$4r^2_j+4r^2_j=(2r_j+2r_i)^2$ $\Longrightarrow$ $\sqrt{8}r_j=2r_j+2r_i$ $\Longrightarrow$
			\begin{equation}
			r_i=\dfrac{\sqrt{8}-2}{2}r_j=\gamma r_j
			\end{equation}
			
			From equations (1) and (2), we get:\\
			$\dfrac{x_{UAV}}{x_b+x_{UAV}}=\dfrac{\sqrt{8}-2}{2}$ $\Longrightarrow$ $2x_{UAV}=x_b(\sqrt{8}-2)+x_{UAV}(\sqrt{8}-2)$
			$\Longrightarrow$\\ $x_{UAV}=x_b\dfrac{(\sqrt{8}-2)}{(4-\sqrt{8})}=0.7071x_b$
			\qedhere
		\end{proof}
		Thus, to guarantee the upside-down arrangements of the truncated cones, we must place the UAVs at horizontal distance equals to $0.7071x_b$. Theorems 2 and 3 help us to transform the geometric problem from 3D to 2D and present an efficient algorithm that maximizes the indoor wireless coverage.	
		\begin{figure*}[t]
			\begin{minipage}[b]{0.45\linewidth}
				\centering
				\includegraphics[width=\textwidth]{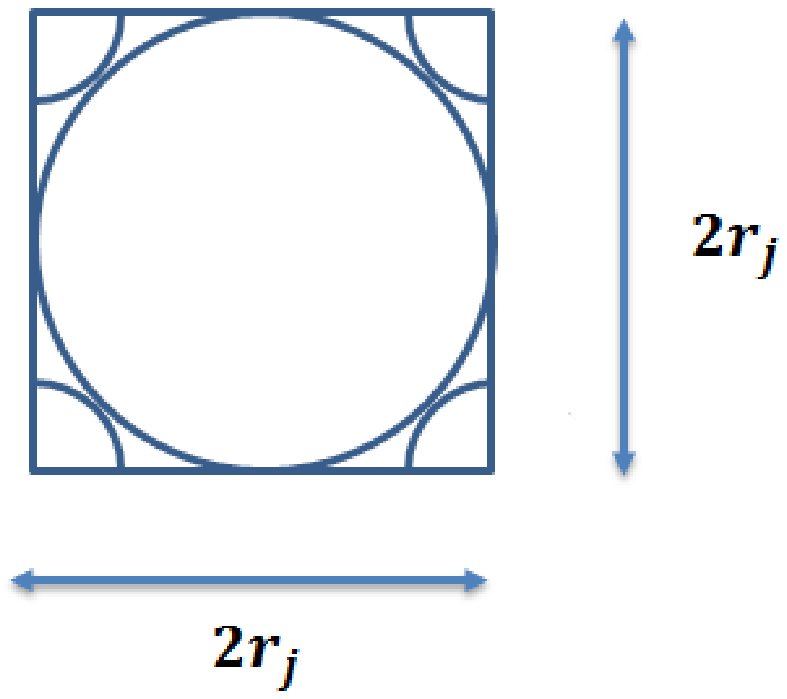}
				\caption{The square cell in side A
				}
				\label{fig10}
			\end{minipage}
			\hspace{1.5cm}
			\begin{minipage}[b]{0.43\linewidth}
				\centering
				\includegraphics[width=\textwidth]{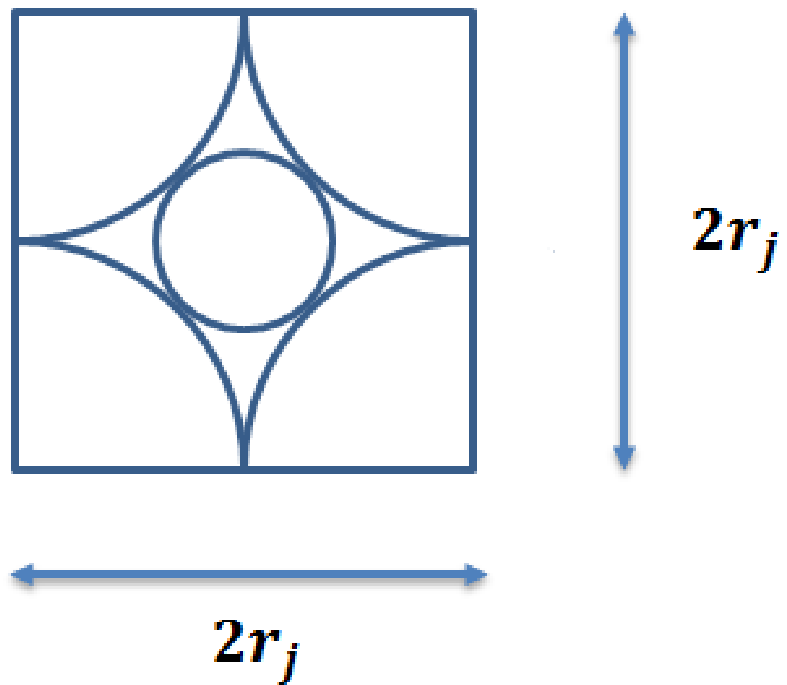}
				\caption{The square cell in side B
				}
				\label{fig11}
			\end{minipage}
		\end{figure*}
		\begin{figure*}[t]
			\begin{minipage}[b]{0.3\linewidth}
				\centering
				\includegraphics[width=\textwidth]{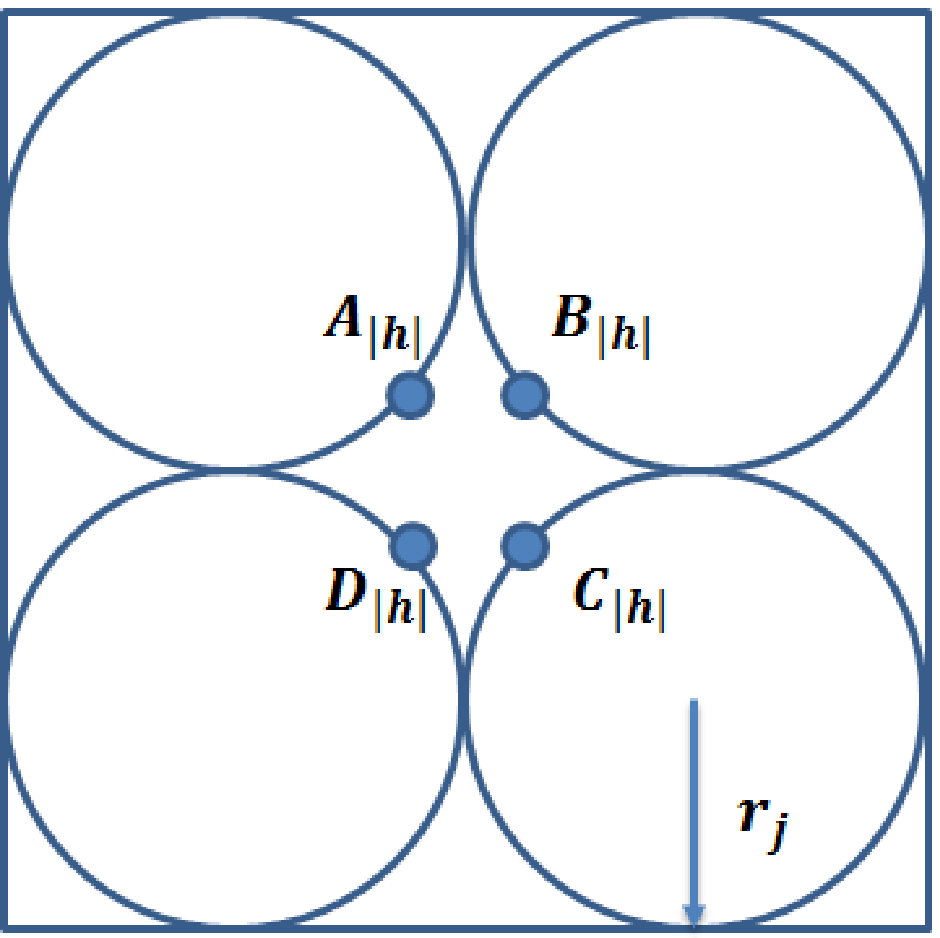}
				\caption{Four circles (with radius $r_j$) in building side $A$
				}
				\label{fig12}
			\end{minipage}
			\hspace{4cm}
			\begin{minipage}[b]{0.3\linewidth}
				\centering
				\includegraphics[width=\textwidth]{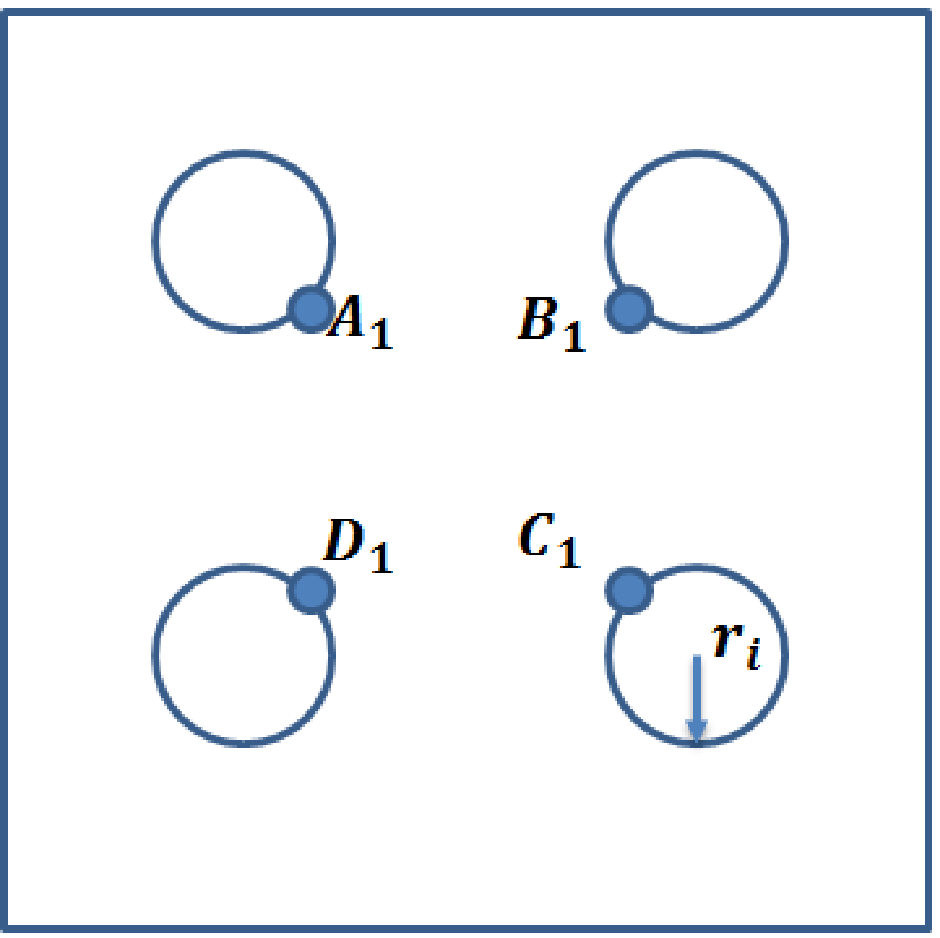}
				\caption{Four circles (with radius $r_i$) in building side $B$
				}
				\label{fig13}
			\end{minipage}
		\end{figure*}
		\begin{theorem}
			The truncated cones do not intersect in 3D iff The circles do not intersect in building sides (A and B).
		\end{theorem}
		\begin{proof}
			First, we prove that if the truncated cones do not intersect in $3D$, then the circles do not intersect in building sides (A and B). Assume that we have a set of truncated cones $G=\{1,2,...,N\}$ and they do not intersect in 3D space. Each truncated cone $n \in G$ can be represented by a number of $2D$ circles $\{c_{1n}, c_{2n},..., c_{|h|n}\}$, where $|h|$ is the height of the truncated cone, $c_{1n}$ is the smaller circular face and $c_{|h|n}$ is the larger circular face. It is obvious that if the $|G|$ truncated cones do not intersect in $3D$ space then the smaller and larger circular faces do not intersect in building sides ($A$ and $B$).
			
			Second, we prove that if the circles do not intersect in building sides (A and B), then the truncated cones do not intersect in $3D$. Assume that four circles (with large radius $r_j$) not intersect in building side $A$ (see Figure~\ref{fig12}), then the circles (with small radius $r_i$) in building side $B$ will appear as shown Figure~\ref{fig13}. Now, we need to do two steps: 1) Connect the lines between these points ($A_{|h|}$ with $A_1$, $B_{|h|}$ with $B_1$, $C_{|h|}$ with $C_1$ and $D_{|h|}$ with $D_1$ ). 2) Draw circles that pass through four points $A_k$, $B_k$, $C_k$ and $D_k$ where $k\in h$. After these two steps, the circles that have been drawn in step two will represent a truncated cone that his circular bases do not intersect with the four circles in building sides ($A$ and $B$). Also, the truncated cones do not intersect in 3D space.
			\qedhere
		\end{proof}
\begin{theorem}
	We maximize the percentage of covered area of building sides (A and B) iff We maximize the percentage of covered volume of building
\end{theorem}
\begin{proof}
	First, we divide the building sides $A$ and $B$ to square cells (as shown in Figures~\ref{fig10} and~\ref{fig11}). The percentage of covered volume is given by:

	\begin{equation}
	 V=\dfrac{\lfloor\dfrac{(y_b\ast z_b)}{4r^2_j}\rfloor\ast 2 \ast (\frac{\pi}{3}\ast x_b\ast (r^2_i+r_ir_j+r^2_j))}{(x_b\ast y_b\ast z_b)}
	 \end{equation}
	Where:\\
	$\lfloor\dfrac{(y_b\ast z_b)}{4r^2_j}\rfloor$: the number of square cells in the building side.\\
	2: the number of truncated cones in the square cell (see Figures~\ref{fig7} and~\ref{fig8}).\\
    $\frac{\pi}{3}\ast x_b\ast (r^2_i+r_ir_j+r^2_j)$: the volume of truncated cone.\\
    $(x_b\ast y_b\ast z_b)$: the volume of the building.
    
   Now, from equations (2) and (3), we get:
   \begin{equation}
   \begin{split}
    V=\dfrac{\lfloor\dfrac{(y_b\ast z_b)}{4r^2_j}\rfloor \ast (\frac{2\pi}{3})\ast (\gamma^2+\gamma+1)r^2_j}{(y_b\ast z_b)}
    =
    K_1\lfloor\dfrac{(y_b\ast z_b)}{4r^2_j}\rfloor r^2_j 
       \end{split}
   \end{equation}
   Where:\\
   $K_1=\dfrac{(\frac{2\pi}{3})(\gamma^2+\gamma+1)}{(y_b\ast z_b)}$\\
   
   The percentage of covered area of building sides ($A$ and $B$) is given by:
   \begin{equation}
   \begin{split}
   W= \dfrac{\lfloor\dfrac{(y_b\ast z_b)}{4r^2_j}\rfloor\ast (\pi r^2_i+\pi r^2_j)}{(y_b\ast z_b)} + \dfrac{\lfloor\dfrac{(y_b\ast z_b)}{4r^2_j}\rfloor\ast (\pi r^2_i+\pi r^2_j)}{(y_b\ast z_b)}=  \dfrac{\lfloor\dfrac{(y_b\ast z_b)}{4r^2_j}\rfloor\ast 2\pi( r^2_i+ r^2_j)}{(y_b\ast z_b)}
      \end{split} 
   \end{equation}
   Now, from equations (2) and (5), we get:
   \begin{equation}
   W=\dfrac{\lfloor\dfrac{(y_b\ast z_b)}{4r^2_j}\rfloor\ast 2\pi(\gamma^2+1)r^2_j}{(y_b\ast z_b)}=K_2\lfloor\dfrac{(y_b\ast z_b)}{4r^2_j}\rfloor r^2_j~~~
   \end{equation}
   Where:\\
   $K_2=\dfrac{(2\pi)(\gamma^2+1)}{(y_b\ast z_b)}$\\
   	To prove that maximizing the percentage of covered volume of building is equivalent to maximizing the percentage of covered area of building sides (A and B). From equations (4) and (6), maximizing $V=K_1\lfloor\dfrac{(y_b\ast z_b)}{4r^2_j}\rfloor r^2_j$ is equivalent to maximizing $K_2\lfloor\dfrac{(y_b\ast z_b)}{4r^2_j}\rfloor r^2_j$ where $K_1$ and $K_2$ are constants.
   	
   	To prove that maximizing the percentage of covered area of building sides (A and B) is equivalent to maximizing the percentage of covered volume of building. From equations (4) and (6), maximizing $W=K_2\lfloor\dfrac{(y_b\ast z_b)}{4r^2_j}\rfloor r^2_j$ is equivalent to maximizing $K_1\lfloor\dfrac{(y_b\ast z_b)}{4r^2_j}\rfloor r^2_j$ where $K_1$ and $K_2$ are constants.
	\qedhere
	\end{proof}
	
In Algorithm 2, we maximize the covered volume by placing the UAVs
in alternating upside-down arrangements. First, we find the horizontal distance between the building and the UAVs $x_{UAV}=0.7071x_b$ (see Theorem 1) that guarantees the alternating upside-down arrangements. Then, we divide the building sides $A$ and $B$ to square cells and place one UAV in front of the square cell. In steps (8-16), we find the $3D$ locations of UAVs that cover the building from side B. On the other hand, steps (17-25) find the $3D$ locations of UAVs that cover the building from side A. Finally, the algorithm will output total number of UAVs and the total covered volume.
\begin{algorithm}
	\begin{algorithmic}
		\STATE 1: \textbf{Input:}
		\STATE 2: The dimensions of building $x_b$, $y_b$ and $z_b$
		\STATE 3: The radius of the larger circular face $r_j$
	    \STATE 4: \textbf{Initialization:}
	    \STATE 5: $r_i=\dfrac{\sqrt{8}-2}{2}r_j$
	    \STATE 6: $x_{UAV}=0.7071x_b$
	    \STATE 7: $u=q=0$
	    \STATE 8: The $3D$ locations of UAVs that cover the building from 
	    \STATE ~~~~side B are given by:
	    \STATE 9: \textbf{For} $k_1=1:\lfloor\dfrac{y_b}{2r_j}\rfloor$
	    \STATE 10: ~~~~~\textbf{For} $s_1=1:\lfloor\dfrac{z_b}{2r_j}\rfloor$
	    \STATE 11: ~~~~~~~~~~ $u=u+1$
	    \STATE 12: ~~~~~~~~~~ $x_q=x_{UAV}+x_b$
	    \STATE 13: ~~~~~~~~~~ $y_u=(2k_1-1)r_j$
	    \STATE 14: ~~~~~~~~~~ $z_u=(2s_1-1)r_j$
	    \STATE 15: ~~~~~\textbf{End}
	    \STATE 16: \textbf{End}
	    \STATE 17: The $3D$ locations of UAVs that cover the building from 
	    \STATE~~~~~side A are given by:
	    \STATE 18: \textbf{For} $k_2=1:\lfloor\dfrac{y_b}{3r_j}\rfloor$
	    \STATE 19: ~~~~~\textbf{For} $s_2=1:\lfloor\dfrac{z_b}{3r_j}\rfloor$
	    \STATE 20: ~~~~~~~~~~ $q=q+1$
	    \STATE 21: ~~~~~~~~~~ $x_q=-x_{UAV}$
	    \STATE 22: ~~~~~~~~~~ $y_q=(2k_2)r_j$
	    \STATE 23: ~~~~~~~~~~ $z_q=(2s_2)r_j$
	    \STATE 24: ~~~~~\textbf{End}
	    \STATE 25: \textbf{End}
		\STATE 26: \textbf{Output:}
		\STATE 27: The number of UAVs$=u+q+2k_1+2s_1$
		\STATE 28: The covered volume=$(u)(2)(\frac{\pi}{3}\ast x_b\ast (r^2_i+r_ir_j+r^2_j))$
	\end{algorithmic}
	\caption{Maximizing Indoor Wireless Coverage Using UAVs}
\end{algorithm}

\section{SIMULATION RESULTS}
Let the dimensions of the building, in the shape of a rectangular prism, be $[0, x_b=30]\times[0, y_b=40]\times[0, zb=60]$. We use three methods to cover the building using UAVs. In the first method, we place all UAVs in front of one building side ($A$ or $B$) (FOBS). In the second method, we place all UAVs above the building ($C$) (ABS). In the third method, we arrange the UAVs in alternating upside-down arrangements (AUDA). For the first and second methods, we utilize the circle packing in a rectangle approach~\cite{PackingProblems} to maximize the covered volume. For the third method, we apply Algorithm 2 to maximize the covered volume. In Figure~\ref{fig14}, we find the maximum total coverage for different antenna half power beamwidth angles $\theta_B$. As can be seen from the simulation results, the maximum total coverage is less than half for the FOBS and ABS methods, this is because providing wireless coverage from one building side will only maximize the covered area of the building side. On the other hand, we improve the maximum total coverage by applying the AUDA, this is because AUDA will allow us to use a higher number of UAVs to provide wireless coverage compared with providing wireless coverage from one building side, as shown in Figure~\ref{fig15}.
  
\begin{figure}[t]
	\centering
	\includegraphics[scale = 0.58]{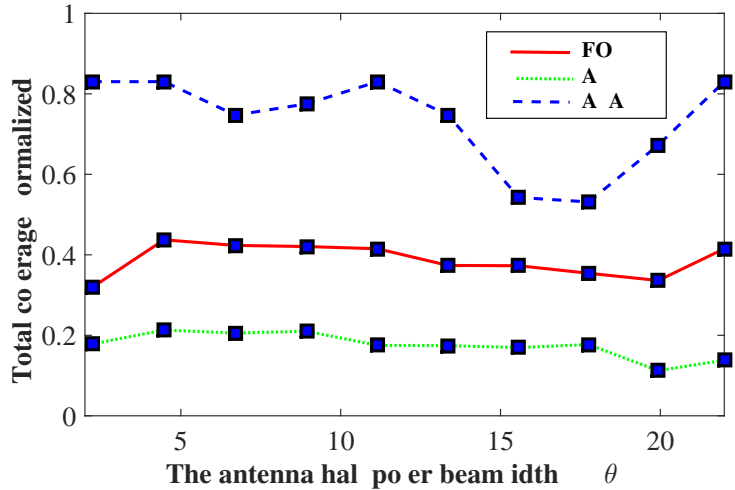}
	\caption{Total coverage vs. $\theta_B$}
	\label{fig14}
\end{figure}
\begin{figure}[t]
	\centering
	\includegraphics[scale = 0.58]{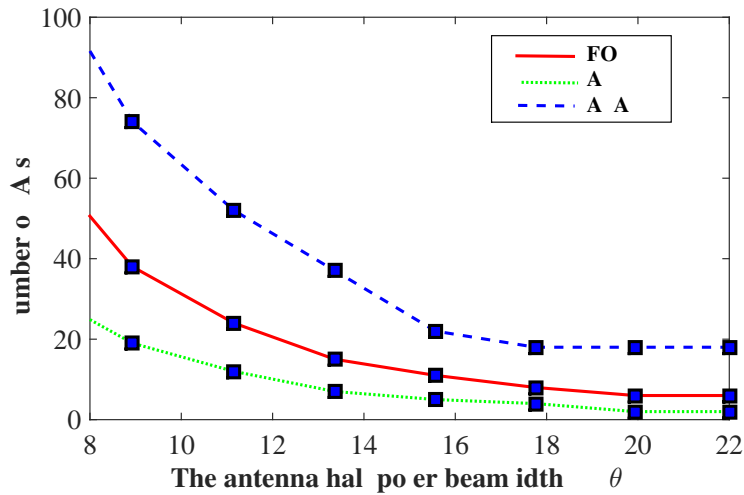}
	\caption{Number of UAVs vs. $\theta_B$}
	\label{fig15}
\end{figure}
In order to provide full wireless coverage for the building, we use UAVs with different channels to cover the holes in the building. In Figure~\ref{fig16}, we find the total number of UAVs required to provide full coverage. As can be seen from the figure, FOBS and ABS need high number of UAVs to guarantee full wireless coverage for the building, due to the irregular shapes of the holes in the building. Here, we can easily specify the number of UAVs required to cover each hole in the building, due to the small projections of the holes in the building side. On the other hand, AUDA needs fewer number of UAVs to provide full wireless coverage, due to the small-regular shapes of the uncovered spaces inside the building. Here, we need only one UAV to cover each hole. In Figure~\ref{fig17}, we find the total transmit power consumed by UAVs when the building is fully covered. Here, we assume that the threshold SNR equals 25dB, the noise power equals -120dBm, the frequency of the channel is 2GHz and the antenna gain of each indoor user is 14.4 dB~\cite{feick2015achievable}. As can be seen from the figure, the total transmit power in all methods is very small, due to the high gain of the directional antennas. Also, we can notice that the total power consumed in FOBS and ABS is higher than that of AUDA. This is because the number of UAVs required to fully cover the building in AUDA is fewer than that for FOBS and ABS.
\begin{figure}[t]
	\centering
	\includegraphics[scale = 0.58]{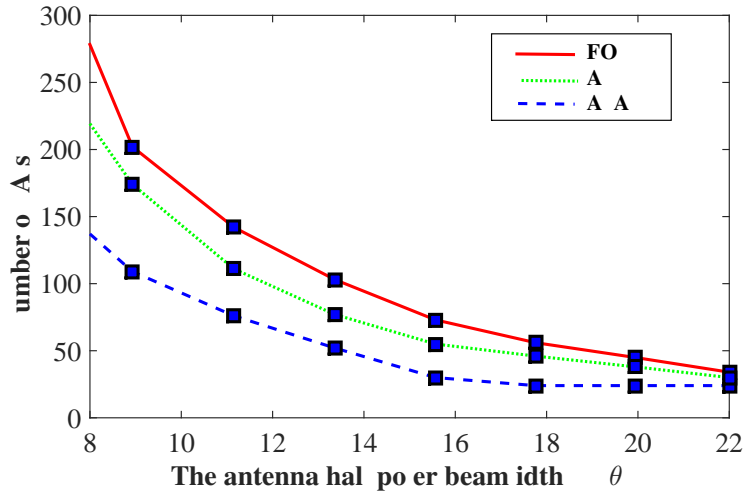}
	\caption{Number of UAVs vs. $\theta_B$}
	\label{fig16}
\end{figure}
\begin{figure}[t]
	\centering
	\includegraphics[scale = 0.58]{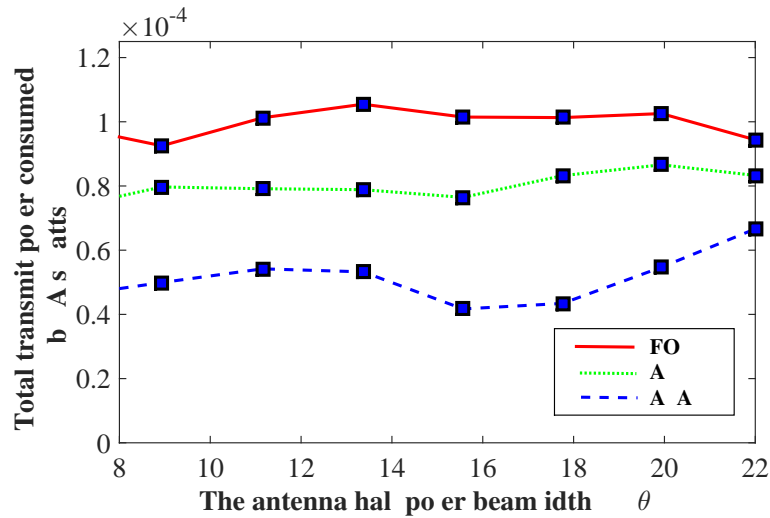}
	\caption{Total transmit power vs. $\theta_B$}
	\label{fig17}
\end{figure}
\section{Conclusion}
Choosing the appropriate placement of UAVs will be a
critical issue when we aim to maximize the indoor wireless coverage. In this paper, we study the case that the UAVs are using one channel, thus in order to maximize the total indoor wireless coverage, we avoid any overlapping in their coverage volumes. We present two methods to place the UAVs; providing wireless coverage from one building side and from two building sides. In the first method, we utilize circle packing theory to determine the 3-D locations of the UAVs in a way that the total coverage area is maximized. In the second method, we place the UAVs in front of two building sides and efficiently arrange the UAVs in alternating upside-down arrangements. We show that the upside-down arrangements problem can be transformed from 3D to 2D and based on that we present an efficient algorithm to solve the problem. Our results show that the upside-down arrangements, can improve the maximum total coverage by 100\% compared to providing wireless coverage from one building side.
\section*{Acknowledgment}
This work was supported in part by the NSF under Grant
CNS-1647170.
	\bibliographystyle{IEEEtran}
	\bibliography{UAVpath}
	
\end{document}